\documentclass[submission,copyright,creativecommons]{eptcs}
\providecommand{\event}{MFPS 2021} 
\usepackage{breakurl}             
\usepackage{underscore}           

\usepackage{mathtools}
\usepackage{stmaryrd} 
\usepackage{enumitem} 
\usepackage{manfnt} 
\usepackage{marginnote} 

\usepackage{amsthm}
\usepackage{amssymb}

\newtheorem{theorem}{Theorem}[section]
\newtheorem{lemma}[theorem]{Lemma}
\newtheorem{proposition}[theorem]{Proposition}
\newtheorem{corollary}[theorem]{Corollary}

\theoremstyle{definition}
\newtheorem{definition}[theorem]{Definition}
\newtheorem{example}[theorem]{Example}

\newtheorem{warning}[theorem]{Warning}

\DeclarePairedDelimiter{\pa}{(}{)}
\DeclarePairedDelimiter{\set}{\{}{\}}

\newcommand{\Nat}{\mathbb N}
\newcommand{\Rat}{\mathbb Q}
\newcommand{\Rea}{\mathbb R}
\newcommand{\Two}{\mathbf{2}}

\newcommand{\Sierp}{\mathbb S}

\newcommand{\Seqspace}{\mathbf A}
\newcommand{\Seqdom}{\mathcal A}
\newcommand{\Finseq}{A^\ast}

\newcommand{\Realdom}{\mathcal R}
\newcommand{\LowerReal}{\mathcal L}

\newcommand{\below}{\mathrel{\sqsubseteq}}
\newcommand{\sbelow}{\mathrel{\sqsubset}}
\newcommand{\aboveorder}{\mathrel{\sqsupseteq}}
\newcommand{\dirsup}{\bigsqcup}
\newcommand{\glb}{\mathrel{\sqcap}}

\DeclareMathOperator{\dset}{\downarrow}
\DeclareMathOperator{\upset}{\uparrow}
\newcommand\twoheaddownarrow{\rotatebox[origin=c]{270}{\(\twoheadrightarrow\)}}
\DeclareMathOperator{\ddset}{\twoheaddownarrow}
\newcommand\twoheaduparrow{\rotatebox[origin=c]{90}{\(\twoheadrightarrow\)}}
\DeclareMathOperator{\upupset}{\twoheaduparrow}

\DeclarePairedDelimiter{\steppa}{\llparenthesis}{\rrparenthesis}

\newcommand{\apart}{\mathrel{\#}}
\newcommand{\posnotbelow}{\mathrel{\not\!\not{\!\sqsubseteq}}}
\newcommand{\apartvar}{\mathrel{\sharp}}

\newcommand{\lcompl}{{\lnot}\hspace{0.12em}}
\newcommand{\compl}{{\sim}}
\newcommand{\acompl}{-}
\newcommand{\setapart}{\mathrel{\bowtie}}

\newcommand{\initseg}[2]{\bar{#1}_{#2}}

\DeclareMathOperator{\Idl}{Idl}
\newcommand{\powerset}[1]{\mathcal P\pa*{#1}}
\newcommand{\To}{\Rightarrow}
\newcommand{\interior}[1]{\pa*{#1}^\circ}
\newcommand{\refined}{\mathrel{\upupset}}

\title{Sharp Elements and Apartness in Domains}
\author{Tom de Jong\thanks{A version of this paper containing full proofs can be found here:
    \href{https://arxiv.org/abs/2106.05064}{\texttt{arXiv:2106.05064}}.}
  \institute{School of Computer Science\\
    University of Birmingham\\
    Birmingham, United Kingdom}
  \email{t.dejong@pgr.bham.ac.uk}
}
\def\titlerunning{Sharp Elements and Apartness in Domains}
\def\authorrunning{T. de Jong}

\begin{document}
\maketitle

\begin{abstract}
  Working constructively, we study continuous directed complete posets (dcpos)
  and the Scott topology.
  Our two primary novelties are a notion of intrinsic apartness and a notion of
  sharp elements.
  Being apart is a positive formulation of being unequal, similar to how
  inhabitedness is a positive formulation of nonemptiness.
  To exemplify sharpness, we note that a lower real is sharp if and only if it
  is located.
  Our first main result is that for a large class of continuous dcpos, the
  Bridges--{V\^i\c{t}\v{a}} apartness topology and the Scott topology coincide.
  Although we cannot expect a tight or cotransitive apartness on nontrivial
  dcpos, we prove that the intrinsic apartness is both tight and cotransitive
  when restricted to the sharp elements of a continuous dcpo.
  These include the strongly maximal elements, as studied by Smyth and
  Heckmann. We develop the theory of strongly maximal elements highlighting its
  connection to sharpness and the Lawson topology.
  Finally, we illustrate the intrinsic apartness, sharpness and strong
  maximality by considering several natural examples of continuous dcpos:
  the Cantor and Baire domains, the partial Dedekind reals and the lower reals.
\end{abstract}

\section{Introduction}\label{sec:introduction}
Domain theory~\cite{AbramskyJung1995} is rich with applications in semantics of
programming languages~\cite{Scott1993,Scott1982,Plotkin1977}, topology and
algebra~\cite{GierzEtAl2003}, and higher-type
computability~\cite{LongleyNormann2015}. The basic objects of domain theory are
directed complete posets (dcpos), although we often restrict our attention to
algebraic or continuous dcpos which are generated by so-called compact elements
or, more generally, by the so-called way-below relation
(Section~\ref{sec:preliminaries}).
We examine the Scott topology on dcpos using an apartness relation and a notion
of sharp elements. Our work is constructive in the sense that we do not assume
the principle of excluded middle or choice axioms, so our results are valid in
any elementary topos.

Classically, i.e.\ when assuming excluded middle, a dcpo with the Scott topology
satisfies \(T_0\)-separation: if two points have the same Scott open
neighbourhoods, then they are equal.
This holds constructively if we restrict to continuous dcpos.
A~classically equivalent formulation of \(T_0\)-separation is: if \(x \neq y\),
then there is a Scott open separating \(x\) and \(y\), i.e.\ containing \(x\)
but not \(y\) or vice versa.
This second formulation is equivalent to excluded middle.
This brings us to the first main notion of this paper
(Section~\ref{sec:intrinsic-apartness}). We say that \(x\)~and~\(y\) are
intrinsically apart, written \(x \apart y\), if there is a Scott open containing
\(x\) but not \(y\) or vice versa. Then \(x \apart y\) is a positive formulation
of \(x \neq y\), similar to how inhabitedness (i.e.\ \(\exists\,{x \in X}\)) is
a positive formulation of nonemptiness (i.e.\ \(X \neq \emptyset\)).

This definition works for any dcpo, but the intrinsic apartness is mostly of
interest to us for continuous dcpos. In fact, the apartness really starts to
gain traction for continuous dcpos that have a basis satisfying certain
decidability conditions.
For example, we prove that for such continuous dcpos, the
Bridges--{V\^i\c{t}\v{a}} apartness topology~\cite{BridgesVita2011} and the
Scott topology coincide (Section~\ref{sec:apartness-topology}).
Thus our work may be regarded as showing that the constructive framework by
Bridges and V\^i\c{t}\v{a} is applicable to domain theory.
It should be noted that these decidability conditions are satisfied by the major
examples in applications of domain theory to topology and computation. Moreover,
these conditions are stable under products of dcpos and, in the case of bounded
complete algebraic dcpos, under exponentials (Section~\ref{sec:preliminaries}).

In \cite[p.~7]{BridgesRichman1987}, \cite[p.~8]{MinesRichmanRuitenburg1988} and
\cite[p.~8]{BridgesVita2011}, an irreflexive and symmetric relation is called an
{inequality (relation)} and the symbol~\({\neq}\) is used to denote it. In
\cite[Definition~2.1]{BishopBridges1985}, an inequality is moreover required to
be cotransitive:
\begin{center}
  if \(x \neq y\), then \(x \neq z\) or \(y \neq z\) for any
    \(x\), \(y\) and \(z\).
\end{center}
The latter is called a {preapartness} in
\cite[Section~8.1.2]{TroelstraVanDalen1988} and the symbol~\({\apart}\) is used
to denote it, reserving \({\neq}\) for the logical negation of equality and the
word {apartness} for a relation that is also tight: if \(\lnot(x \apart y)\),
then \(x = y\).

\begin{warning}\label{warning}
  \marginnote{\dbend}
  We deviate from the above and use the word apartness and the symbol
  \({\apart}\) for an irreflexive and symmetric relation, so we do not require
  it to be cotransitive or tight.
\end{warning}
The reasons for our choice of terminology and notations are as follows: (i) we
wish to reserve \(\neq\) for the negation of equality as in
\cite[Section~8.1.2]{TroelstraVanDalen1988}; (ii) the word inequality is
confusingly also used in the context of posets to refer to the partial order;
and finally, (iii) the word inequality seems to suggest that the negation of the
inequality relation is an equivalence relation, but, in the absence of
cotransitivity, it need not be.

Actually, we prove that no apartness on a nontrivial dcpo can be
cotransitive or tight unless (weak) excluded middle holds.
However, there is a natural collection of elements for which the intrinsic
apartness is both tight and apartness: the sharp elements
(Section~\ref{sec:tightness-cotransitivity-sharpness}). Sharpness is slightly
involved in general, but it is easy to understand for algebraic dcpos: an
element \(x\) is sharp if and only if for every compact element it is decidable
whether it is below \(x\).
Moreover, the notion is quite natural in many examples. For instance, the sharp
elements of a powerset are exactly the decidable subsets and the sharp lower
reals are precisely the located ones.

An import class of sharp elements is given by the strongly maximal elements
(Section~\ref{sec:strongly-maximal-elements}). These were studied in a classical
context by Smyth~\cite{Smyth2006} and Heckmann~\cite{Heckmann1998}, because of
their desirable properties. For instance, while the subspace of maximal elements
may fail to be Hausdorff, the subspace of strongly maximal elements is both
Hausdorff (two distinct points can be separated by disjoint Scott opens) and
regular (every neighbourhood contains a Scott closed neighbourhood).
As shown in \cite{Smyth2006}, strong maximality is closely related to the Lawson
topology. Specifically, Smyth proved that a point \(x\) is strongly maximal
if and only if every Lawson neighbourhood of \(x\) contains a Scott
neighbourhood of \(x\).
Using sharpness, we offer a constructive proof of~this.

Finally, we illustrate (Section~\ref{sec:examples}) the above notions by
presenting examples of continuous dcpos that embed well-known spaces as
strongly maximal elements:
the Cantor and Baire domains, the partial Dedekind reals and the lower reals.

\paragraph{Related work}
There are numerous accounts of basic domain theory in several constructive
systems, such as
\cite{SambinValentiniVirgili1996,Negri1998,Negri2002,MaiettiValentini2004,Kawai2017,Kawai2021}
in the predicative setting of formal topology~\cite{Sambin1987,CoquandEtAl2003},
as well as works in various type theories: \cite{Hedberg1996} in (a version of)
Martin-L\"of Type Theory, \cite{Lidell2020} in Agda,
\cite{BentonKennedyVarming2009,Dockins2014} in Coq and our previous
work~\cite{deJongEscardo2021a} in univalent foundations. Besides that, the
papers~\cite{BauerKavkler2009,PattinsonMohammadian2021} are specifically aimed
at program extraction.

Our work is not situated in formal topology and we work informally in
(impredicative) set theory without using excluded middle or choice axioms. We also
consider completeness with respect to all directed subsets and not just
\(\omega\)-chains as is done
in~\cite{BauerKavkler2009,PattinsonMohammadian2021}. The principal contributions
of our work are the aforementioned notions of intrinsic apartness and sharp
elements, although the idea of sharpness also appears in formal topology: an
element of a continuous dcpo is sharp if and only if its filter of Scott open
neighbourhoods is located in the sense of Spitters~\cite{Spitters2010} and
Kawai~\cite{Kawai2017}.

If, as advocated in \cite{Abramsky1987,Vickers1989,Smyth1993}, we think of
(Scott) opens as observable properties, then this suggests that we label two
points as apart if we have made conflicting observations about them, i.e.\ if
there are disjoint opens separating the points. Indeed, (an equivalent
formulation of) this notion is used in Smyth's \cite[p.~362]{Smyth2006}.
While these notions are certainly useful, both in the presence and absence of
excluded middle, our apartness serves a different purpose: It is a positive
formulation of the negation of equality used when reasoning about the Scott
topology on a dcpo, which (classically) is only a \(T_0\)-space that isn't
Hausdorff in general. By contrast, an apartness based on disjoint opens would
supposedly perform a similar job for a Hausdorff space, such as a dcpo with the
Lawson topology.

Finally, von Plato~\cite{vonPlato2001} gives a constructive account of
so-called positive partial orders: sets with a binary relation \({\not\leq}\)
that is irreflexive and cotransitive (i.e.\ if \(x \not\leq y\), then
\(x \not\leq z\) or \(y \not\leq z\) for any elements \(x\), \(y\) and~\(z\)).
Our notion \({\posnotbelow}\) from Definition~\ref{def:specialization} bears
some similarity, but our work is fundamentally different for two
reasons. Firstly, \({\posnotbelow}\) is not cotransitive. Indeed, we cannot
expect such a cotransitive relation on nontrivial dcpos, cf.\
Theorem~\ref{tight-cotransitive-wem}. Secondly, in \cite{vonPlato2001} equality
is a derived notion from \({\not\leq}\), while equality is primitive for us.

\section{Preliminaries}\label{sec:preliminaries}
We give the basic definitions and results in the theory of (continuous)
dcpos. It is not adequate to simply refer the reader to classical texts on
domain theory~\cite{AbramskyJung1995,GierzEtAl2003}, because two classically
equivalent definitions need not be constructively equivalent, and hence we need
to make choices here. For example, while classically every Scott open subset is
the complement of a Scott closed subset, this does not hold constructively
(Lemma~\ref{complement-of-open-is-closed}).
The results presented here are all provable constructively. Constructive proofs
of standard results, such as Lemma~\ref{way-below-basics} and
Proposition~\ref{interpolation} (the interpolation property), can be found
in~\cite{deJongEscardo2021a}.
Finally, in Section~\ref{sec:decidability-conditions} we introduce and study
some decidability conditions on bases of dcpos that will make several
appearances throughout the paper. These decidability conditions always
hold if excluded middle is assumed.

\begin{definition}\hfill
  \begin{enumerate}
  \item A subset \(S\) of a poset \((X,\below)\) is \emph{directed} if it is
    \emph{inhabited}, meaning there exists \(s \in S\), and \emph{semidirected}:
    for every two points \(x,y \in S\) there exists \(z \in S\) with
    \(x \below z\) and \(y \below z\).
  \item A \emph{directed complete poset (dcpo)} is a poset where every
    directed subset \(S\) has a supremum, denoted by~\(\dirsup S\).
  \item A dcpo is \emph{pointed} if it has a least element, typically denoted by
    \(\bot\).
  \end{enumerate}
\end{definition}

Notice that a poset is a pointed dcpo if and only if it has suprema for all
semidirected subsets.
In fact, given a pointed dcpo \(D\) and a semidirected subset \(S \subseteq D\),
we can consider the directed subset \(S \cup \set{\bot}\) of \(D\) whose
supremum is also the supremum of \(S\).

\begin{definition}
  An element \(x\) of a dcpo \(D\) is \emph{way below} an element
  \(y \in D\) if for every directed subset \(S\) with \(y \below \dirsup S\)
  there exists \(s \in S\) such that \(x \below s\) already. We denote this by
  \(x \ll y\) and say \(x\) is \emph{way below} \(y\).
\end{definition}
The following is easily verified.
\begin{lemma}\label{way-below-basics}
  The way-below relation enjoys the following properties:
  \begin{enumerate}
  \item it is transitive;
  \item if \(x \below y \ll z\), then \(x \ll z\) for every \(x\), \(y\) and
    \(z\);
  \item if \(x \ll y \below z\), then \(x \ll z\) for every \(x\), \(y\) and
    \(z\).
  \end{enumerate}
\end{lemma}
\begin{definition}
  A dcpo \(D\) is \emph{continuous} if for every element \(x \in D\),
  the subset \(\ddset x \coloneqq \set*{y \in D \mid y \ll x}\) is directed and
  its supremum is \(x\).
\end{definition}
\begin{definition}
  An element of a dcpo is \emph{compact} if it is way below itself.  A dcpo
  \(D\) is \emph{algebraic} if for every element \(x \in D\), the
  subset \(\set{c \in D \mid c \below x \text{ and \(c\) is
    compact}}\) is directed with supremum \(x\).
\end{definition}
\begin{proposition}
  Every algebraic dcpo is continuous.
\end{proposition}

\subsection{The Scott topology}
\begin{definition}\hfill
  \begin{enumerate}
  \item A subset \(C\) of a dcpo \(D\) is \emph{Scott closed} if it is
    closed under directed suprema and a lower set: if \(x \below y \in C\), then
    \(x \in C\) too.
  \item A subset \(U\) of a dcpo \(D\) is \emph{Scott open} if it is an
    upper set and for every directed subset \(S \subseteq D\) with
    \(\dirsup S \in U\), there exists \(s \in S\) such that \(s \in U\) already.
  \end{enumerate}
\end{definition}
\begin{example}\label{ex:Scott-opens}
  For any element \(x\) of a dcpo \(D\), the subset
  \(\dset x \coloneqq \set{y \in D \mid y \below x}\) is Scott closed.
  If \(D\) is continuous, then the subset
  \(\upupset x \coloneqq \set{y \in D \mid x \ll y}\) is Scott open.
  (One proves this using the interpolation property, which is
  Proposition~\ref{interpolation} below.)
  Moreover, if \(D\) is continuous, then the set
  \(\set{\upupset x \mid x \in X}\) is a basis for the Scott topology on \(D\).
\end{example}
\begin{lemma}\label{complement-of-open-is-closed}
  The complement of a Scott open subset is Scott closed.
  The converse holds if and only if excluded middle does, as we prove in
  Proposition~\ref{complement-of-closed-is-open-em}.
\end{lemma}

\begin{definition}\label{def:interior-closure}
  In a topological space \(X\), the \emph{interior} of a subset
  \(S \subseteq X\) is the largest open of \(X\) contained in \(S\).
  Dually, the \emph{closure} of a subset \(S \subseteq X\) is the smallest
  closed subset of \(X\) that contains \(S\).
\end{definition}

\subsection{(Abstract) bases}
\begin{definition}
  A \emph{basis} for a dcpo \(D\) is a subset \(B \subseteq D\) such that for
  every element \(x \in D\), the subset \({B \cap \ddset{x}}\) is directed with
  supremum \(x\).
\end{definition}

\begin{lemma}
  A dcpo is continuous if and only if it has a basis and a dcpo is algebraic if
  and only if it has a basis of compact elements.
  Moreover, if \(B\) is a basis for an algebraic dcpo \(D\), then \(B\) must
  contain every compact element of \(D\). Hence, an algebraic dcpo has a unique
  smallest basis consisting of compact elements.
\end{lemma}

\begin{example}\label{ex:powerset}
  The powerset \(\powerset{X}\) of any set \(X\) ordered by inclusion and with
  suprema given by unions is a pointed algebraic dcpo. Its compact elements are
  the Kuratowski finite subsets of \(X\).
  A set \(X\) is \emph{Kuratowski finite} if it is finitely enumerable, i.e.\
  there exists a surjection \(e \colon \set{0,\dots,n-1} \twoheadrightarrow X\)
  for some number \(n \in \Nat\).
\end{example}

\begin{definition}\label{def:Sierpinski-domain}
  The \emph{Sierpi\'nski domain} \(\Sierp\) is the free pointed dcpo on a single
  generator.
  We can realize \(\Sierp\) as the set of truth values, i.e.\ as the powerset
  \(\powerset{\set{*}}\) of a singleton. The compact elements of \(\Sierp\) are
  exactly the elements \(\bot \coloneqq \emptyset\) and
  \(\top \coloneqq \set{*}\).
\end{definition}

\begin{lemma}\label{basis-basis}
  For every dcpo \(D\), a subset \(B \subseteq D\) is a basis for the dcpo \(D\)
  if and only if \(\set{\upupset b \mid b \in B}\) is a basis for the Scott
  topology on \(D\).
\end{lemma}

\begin{proposition}\label{basis-is-dense}
  Every basis of a continuous dcpo is dense with respect to the Scott topology
  in the following (classically equivalent) ways:
  \begin{enumerate}
  \item the Scott closure of the basis is the whole dcpo;
  \item every inhabited Scott open contains a point in the basis.
  \end{enumerate}
\end{proposition}

\begin{proposition}[Interpolation]\label{interpolation}
  If \(x \ll y\) are elements of a continuous dcpo \(D\), then there exists
  \(b \in D\) with \(x \ll b \ll y\).
  Moreover, if \(D\) has a basis \(B\), then there exists such an element \(b\)
  in \(B\).
\end{proposition}
\begin{lemma}\label{below-in-terms-of-way-below}
  For every two elements \(x\) and \(y\) of a continuous dcpo \(D\) we have
  \(x \below y\) if and only if \(\forall_{z \in D}\,\pa*{z \ll x \to z \ll y}\).
  Moreover, if \(D\) has a basis \(B\), then
  \(x \below y\) if and only if \(\forall_{b \in B}\,\pa*{b \ll x \to b \ll y}\).
\end{lemma}
\begin{definition}
  An \emph{abstract basis} is a pair \((B,\prec)\) such that \(\prec\) is
  transitive and \emph{interpolative}: for every \(b \in B\), the subset
  \(\dset b \coloneqq \set{a \in B \mid a \prec b}\) is directed.
  The \emph{rounded ideal completion} \(\Idl(B,\prec)\) of an abstract basis
  \((B,\prec)\) consists of directed lower sets of \((B,\prec)\), known as
  \emph{(rounded) ideals}, ordered by subset inclusion. It is a continuous dcpo
  with basis \(\set{\dset b \mid b \in B}\) and directed suprema given by unions.
\end{definition}
\begin{lemma}\label{reflexive-algebraic}
  If the relation \(\prec\) of an abstract basis \((B,\prec)\) is reflexive,
  then \(\Idl(B,\prec)\) is algebraic and its compact elements are exactly those
  of the form \(\dset b\) for \(b \in B\).
\end{lemma}
\subsection{Decidability conditions}\label{sec:decidability-conditions}
Every continuous dcpo \(D\) has a basis, namely \(D\) itself. Our interest in
bases lies in the fact that we can ask a dcpo to have a basis satisfying certain
decidability conditions that we couldn't reasonably impose on the entire
dcpo. For instance, the basis \(\set{\bot,\top}\) of the Sierpi\'nski domain
\(\Sierp\) has decidable equality, but decidable equality on all of \(\Sierp\)
is equivalent to excluded middle.

The first decidability condition that we will consider is for bases \(B\) of a
pointed continuous dcpo:
\begin{equation}\label{equals-bot-decidable}\tag{\(\delta_\bot\)}
  \text{For every } b \in B, \text{ it is decidable whether } b = \bot.
\end{equation}
The second and third decidability conditions are for bases of any continuous dcpo:
\begin{equation}\label{way-below-decidable}\tag{\(\delta_{\ll}\)}
  \text{For every } a,b \in B, \text{ it is decidable whether } a \ll b.
\end{equation}
\begin{equation}\label{below-decidable}\tag{\(\delta_{\below}\)}
  \text{For every } a,b \in B, \text{ it is decidable whether } a \below b.
\end{equation}

Observe that each of \eqref{below-decidable} and \eqref{way-below-decidable}
implies \eqref{equals-bot-decidable} for pointed dcpos, because \(\bot\) is
compact.
In general, neither of the conditions \eqref{below-decidable} and
\eqref{way-below-decidable} implies the other. However, in some cases, for
example when \(B\) is finite, \eqref{way-below-decidable} is a weaker condition,
because of Lemma~\ref{below-in-terms-of-way-below}.
Moreover, if the dcpo is algebraic then conditions \eqref{below-decidable} and
\eqref{way-below-decidable} are equivalent for the unique basis of compact
elements.

Finally, we remark that many natural examples in domain theory satisfy the
decidability conditions. In particular, this holds for all examples in
Section~\ref{sec:examples}.

\begin{proposition}
  Let \(D\) be any pointed dcpo that is nontrivial in the sense that there
  exists \(x \in D\) with \(x \neq \bot\). If \(y = \bot\) is decidable for
  every \(y \in D\), then weak excluded middle follows.
\end{proposition}

Moreover, if the order relation of a dcpo is decidable, then, by antisymmetry,
the dcpo must have decidable equality, but we showed in
\cite[Corollary~39]{deJongEscardo2021b} that this implies (weak) excluded
middle, unless the dcpo is trivial.

\subsubsection{Closure under products and exponentials}
\begin{definition}
  The \emph{product} \(D \times E\) of two dcpos \(D\) and \(E\) is given by
  their Cartesian product ordered pairwise. The supremum of a
  directed subset \(S \subseteq D \times E\) is given by the pair of suprema
  \(\dirsup \set{x \in D \mid \exists_{y \in E}\, (x,y) \in S}\) and
  \(\dirsup \set{y \in D \mid \exists_{x \in D}\, (x,y) \in S}\).
\end{definition}
\begin{proposition}
  If \(D\) and \(E\) are continuous dcpos with bases \(B_D\) and \(B_E\), then
  \(B_D \times B_E\) is a basis for the product \(D \times E\). Also, if
  \(B_D\) and \(B_E\) both satisfy \eqref{equals-bot-decidable}, then so does
  \(B_D \times B_E\), and similarly for \eqref{way-below-decidable} and
  \eqref{below-decidable}.
\end{proposition}

\begin{definition}
  A function between dcpos is \emph{Scott continuous} if it preserves directed
  suprema.
\end{definition}
\begin{definition}
  The \emph{exponential} \(E^D\) of two dcpos \(D\) and \(E\) is given by the
  set of Scott continuous functions from \(D\) to \(E\) ordered pointwise, i.e.\
  \(f \below g\) if \(\forall_{x \in D}\,f(x) \below g(x)\) for
  \(f,g\colon D \to E\). Suprema of directed subsets are also given pointwise.
\end{definition}

\begin{definition}
  Given an element \(x\) of a dcpo \(D\) and an element \(y\) of a pointed dcpo
  \(E\), the \emph{single-step function}
  \(\steppa{x \To y} \colon D \to E\) is defined as
  \(\steppa{x \To y}(d) \coloneqq \dirsup \set{y \mid x \below d}\).
  A \emph{step-function} is the supremum of a Kuratowski finite
  (recall~Example~\ref{ex:powerset}) subset of single-step functions.
\end{definition}

\begin{lemma}\label{step-function-compact}
  If \(x\) is a compact element of \(D\) and \(y\) is any element of a pointed
  dcpo \(E\), then the single-step function
  \(\steppa{x \To y} \colon {D \to E}\) is Scott continuous. If \(y\) is also
  compact, then \(\steppa{x \To y}\) is a compact element of the exponential
  \(E^D\).
\end{lemma}

\begin{lemma}\label{step-function-below}
  For every element \(x\) of a dcpo \(D\), element \(y\) of a pointed dcpo \(E\)
  and Scott continuous function \(f \colon D \to E\), we have
  \(\steppa{x \To y} \below f\) if and only if \(y \below f(x)\).
\end{lemma}

\begin{lemma}\label{compact-closed-under-finite-sups}
  The compact elements of a dcpo are closed under existing finite suprema.
\end{lemma}

\begin{definition}
  A subset \(S\) of a poset \((X,\below)\) is \emph{bounded} if there exists
  \(x \in X\) such that \(s \below x\) for every \(s \in S\).
  A poset \((X,\below)\) is \emph{bounded complete} if every bounded subset
  \(S\) of \(X\) has a supremum \(\dirsup S\) in \(X\).
\end{definition}

\begin{proposition}\label{dec-closed-under-exponentials}
  If \(D\) is an inhabited algebraic dcpo with basis of compact elements \(B_D\)
  and \(E\) is a pointed bounded complete algebraic dcpo with basis of compact
  elements \(B_E\), then
  \begin{align*}
    B \coloneqq \Big{\{}\dirsup S \mid\,\, &S \text{ is a bounded Kuratowski finite
    subset of single-step functions} \\
    &\text{of the form } \steppa{a \To b} \text { with } a \in B_D
      \text{ and } b \in B_E\Big{\}}.
  \end{align*}
  is the basis of compact elements for the algebraic exponential \(E^D\).
  Moreover, if \(B_E\) satisfies \eqref{equals-bot-decidable}, then so
  does~\(B\). Finally, if \(B_D\) and \(B_E\) both satisfy
  \eqref{below-decidable} \textup{(}or equivalently,
  \eqref{way-below-decidable}\textup{)}, then \(B\) satisfies
  \eqref{below-decidable} and \eqref{way-below-decidable} too.
\end{proposition}

\section{The intrinsic apartness}\label{sec:intrinsic-apartness}
\begin{definition}\label{def:specialization}
  The \emph{specialization preorder} on a topological space \(X\) is the
  preorder \({\leq}\) on \(X\) given by putting \(x \leq y\) if every open
  neighbourhood of \(x\) is an open neighbourhood of \(y\). Given \(x,y\in X\),
  we write \(x \posnotbelow y\) and say that \(y\)~\emph{does not
    specialize}~\(x\) if there exists an open neighbourhood of \(x\) that does
  not contain \(y\).
\end{definition}

Observe that \(x \posnotbelow y\) is classically equivalent to \(x \nleq y\),
the logical negation of \(x \leq y\). We also write \(x \not\below y\) for the
logical negation of \(x \below y\).

\begin{definition}\label{def:apartness}
  An \emph{apartness} on a set \(X\) is a binary relation \({\apart}\) on \(X\)
  satisfying
  \begin{enumerate}
  \item \emph{irreflexivity}: \(x \apart x\) is false for every \(x \in X\);
  \item \emph{symmetry}: if \(x \apart y\), then \(y \apart x\) for every
    \(x,y \in X\).
  \end{enumerate}
  If \(x \apart y\) holds, then \(x\) and \(y\) are said to be \emph{apart}.
\end{definition}
Notice that we do not require cotransitivity or tightness, cf.\
Warning~\ref{warning}. Notice that irreflexivity implies that if \(x \apart y\),
then \(x \neq y\), so \(\apart\) is a strengthening of inequality.

\begin{definition}\label{def:intrinsic-apartness}
  Given two points \(x\) and \(y\) of a topological space \(X\), we say that
  \(x\) and \(y\) are \emph{intrinsically apart}, written \(x \apart y\), if
  \(x \posnotbelow y\) or
  \(y \posnotbelow x\).
  Thus, \(x\) is intrinsically apart from \(y\) if there is a Scott open
  neighbourhood of \(x\) that does not contain \(y\) or vice versa.
  It is clear that the relation \({\apart}\) is an apartness in the sense of
  Definition~\ref{def:apartness}.
\end{definition}

With excluded middle, one can show that the specialization preorder for the
Scott topology on a dcpo coincides with the partial order of the dcpo. In
particular, the specialization preorder is in fact a partial
order. Constructively, we still have the following result.

\begin{lemma}\label{order-specialization}
  Let \(x\) and \(y\) be elements of a dcpo \(D\).  If \(x \below y\), then
  \(x \leq y\), where \({\leq}\) is the specialization order of the Scott
  topology. If \(D\) is continuous, then the converse holds too, so \({\below}\)
  and \({\leq}\) coincide in that case.
\end{lemma}

\begin{lemma}\label{posnotbelow-criterion}
  For a continuous dcpo \(D\) we have \(x \posnotbelow y\) if and only if there
  exists \(b \in D\) such that \(b \ll x\), but \(b \not\below y\).
  Moreover, if \(D\) has a basis \(B\), then there exists such an element \(b\)
  in \(B\).
\end{lemma}

The condition in Lemma~\ref{posnotbelow-criterion} appears in a remark right
after \cite[Definition~I-1.6]{GierzEtAl2003}, as a classically equivalent
reading of \(x \not\below y\).

\begin{example}
  Consider the powerset \(\powerset{X}\) of a set \(X\) as a pointed algebraic
  dcpo. Using Lemma~\ref{posnotbelow-criterion}, we see that a subset
  \(A \in \powerset{X}\) is intrinsically apart from the empty set if and only
  if \(A\) is inhabited. More generally, we have \(A \posnotbelow B\) if and
  only if \(B \setminus A\) is inhabited for every two subsets
  \(A,B \in \powerset{X}\).
\end{example}

\begin{proposition}\label{posnotbelow-notbelow}\hfill
  \begin{enumerate}
  \item For any elements \(x\) and \(y\) of a dcpo \(D\), we have that
    \(x \posnotbelow y\) implies \(x \not\below y\).
  \item The converse of~\textup{(i)} holds if and only if excluded middle holds.
    In particular, if the converse of~\textup{(i)} holds for all elements of the
    Sierpi\'nski domain \(\Sierp\), then excluded middle follows.
  \item For any elements \(x\) and \(y\) of a dcpo \(D\), we have that
    \(x \apart y\) implies \(x \neq y\).
  \item The converse of~\textup{(iii)} holds if and only if excluded middle holds.
    In particular, if the converse of~\textup{(iii)} holds for all elements of the
    Sierpi\'nski domain \(\Sierp\), then excluded middle follows.
  \item If \(c\) is a compact element of a dcpo \(D\) and \(x \in D\), then
    \(c \not\below x\) implies \(c \posnotbelow x\), without the need to assume
    excluded middle.
  \end{enumerate}
\end{proposition}

With excluded middle, complements of Scott closed subsets are Scott open. In
particular, the subset \(\set{x \in D \mid x \not\below y}\) is Scott open for
any element \(y\) of a dcpo \(D\). Constructively, we have the following result.
\begin{proposition}\label{interior-of-complement}
  For any element \(y\) of a dcpo \(D\), the Scott interior of
  \(\set{x \in D \mid x \not\below y}\) is given by the subset
  \(\set{x \in D \mid x \posnotbelow y}\), where we recall that
  \(x \posnotbelow y\) means that there exists a Scott open containing \(x\)
  but not \(y\).
\end{proposition}

\begin{proposition}\label{complement-of-closed-is-open-em}
  If the complement of every Scott closed subset of the Sierpi\'nski domain
  \(\Sierp\) is Scott open, then excluded middle follows.
\end{proposition}

An element \(x\) of a pointed dcpo may be said to be nontrivial if
\(x \neq \bot\). Given our notion of apartness, we might consider the
constructively stronger \(x \apart \bot\). We show that this is related to
Johnstone's notion of positivity~\cite[p.~98]{Johnstone1984}.
In~\cite[Definition~25]{deJongEscardo2021b} we adapted Johnstone's positivity
from locales to a general class of posets that includes dcpos. Here we give an
equivalent, but simpler, definition just for pointed dcpos.
\begin{definition}
  An element \(x\) of a pointed dcpo \(D\) is \emph{positive} if every
  semidirected subset \(S \subseteq D\) satisfying \(x \below \dirsup S\) is
  inhabited (and hence directed).
\end{definition}

\begin{proposition}
  For every element \(x\) of a pointed dcpo, if \(x \apart \bot\), then \(x\) is
  positive.
  In the other direction, if \(D\) is a continuous pointed dcpo with a basis
  satisfying \eqref{equals-bot-decidable}, then every positive element of \(D\)
  is apart from~\(\bot\).
\end{proposition}

\section{The apartness topology}\label{sec:apartness-topology}
In~\cite[Section~2.2]{BridgesVita2011}, Bridges and V\^i\c{t}\v{a} start with a
topological space \(X\) equipped with an apartness relation~\({\apart}\) and,
using the topology and apartness, define a second topology on \(X\), known as
the \emph{apartness topology}. A~natural question is whether the original
topology and the apartness topology coincide. For example, if \(X\) is a metric
space and we set two points of \(X\) to be apart if their distance is strictly
positive, then the Bridges-V\^i\c{t}\v{a} apartness topology and the topology
induced by the metric coincide~\cite[Proposition~2.2.10]{BridgesVita2011}.
We show, assuming a modest \(\lnot\lnot\)-stability condition that holds in
examples of interest, that the Scott topology on a continuous dcpo with the
intrinsic apartness relation coincides with the apartness topology.
We start by repeating some basic definitions and results of Bridges and
V\^i\c{t}\v{a}.  Recalling Warning~\ref{warning}, we remind the reader familiar
with~\cite{BridgesVita2011} that Bridges and V\^i\c{t}\v{a} write \(\neq\) and
use the word \emph{inequality} for what we denote by \(\apart\) and call
\emph{apartness}.

In constructive mathematics, positively defined notions are usually more useful
than negatively defined ones. We already saw examples of this: \({\apart}\)
versus \({\neq}\) and \({\posnotbelow}\) versus \({\not\below}\). We now use an
apartness to give a positive definition of the complement of a set.
\begin{definition}
  Given a subset \(A\) of a set \(X\) with an apartness \(\apart\) we define the
  \emph{logical complement} and the \emph{complement} respectively as
  \begin{enumerate}
  \item \(\lcompl A \coloneqq \set{x \in X \mid x \not\in A} =
    \set{x \in X \mid \forall_{y \in A}\, x \neq y}\);
  \item \(\compl A \coloneqq \set{x \in X \mid \forall_{y \in A}\,x \apart y}\).
  \end{enumerate}
\end{definition}

\begin{definition}
  A topological space \(X\) equipped with an apartness \({\apart}\) satisfies
  the \emph{(topological) reverse Kolmogorov property} if for every open \(U\)
  and points \(x,y \in X\) with \(x \in U\) and \(y \not\in U\), we have
  \(x \apart y\).
\end{definition}

\begin{lemma}[Proposition~2.2.2 in \cite{BridgesVita2011}]
  \label{reverse-Kolmogorov-consequence}
  If a topological space \(X\) equipped with an apartness \(\apart\) satisfies
  the reverse Kolmogorov property, then for every subset \(A \subseteq X\), we
  have \(\interior{\lcompl A} = \interior{\compl A}\), where \(Y^\circ\)
  denotes the interior of \(Y\) in \(X\).
\end{lemma}

\begin{definition}
  For an element \(x\) of a topological space \(X\) with an apartness \(\apart\)
  and a subset \(A \subseteq X\), we write \(x \setapart A\) if
  \(x \in \interior{\compl A}\). This gives rise to the \emph{apartness
    complement}: \(\acompl A \coloneqq \set{x \in X \mid x \bowtie A}\).
  Subsets of the form \(\acompl A\) are called \emph{nearly open}.
  The \emph{apartness topology} on \(X\) is the topology whose basic opens are
  the nearly open subsets of~\(X\).
\end{definition}

\begin{lemma}[Proposition~2.2.7 in \cite{BridgesVita2011}]
  \label{nearly-open-is-open}
  Every nearly open subset of \(X\) is open in the original topology of~\(X\).
\end{lemma}

The following are original contributions.
\begin{definition}
  Say that a basis \(B\) for a topological space \(X\) is
  \emph{\(\lnot\lnot\)-stable} if \(U = \interior{\lnot\lnot U}\) for every open
  \(U \in B\).
  Note that \(U \subseteq \interior{\lnot\lnot U}\) holds for every open \(U\),
  so the relevant condition is that \(\interior{\lnot\lnot U} \subseteq U\) for
  every basic open \(U\).
\end{definition}
Examples of such bases will be provided by Theorem~\ref{apartness-Scott-topology} below.

\begin{lemma}\label{matching-topologies-criterion}
  If a topological space \(X\) equipped with an apartness \(\apart\) satisfies
  the reverse Kolmogorov property and has a \(\lnot\lnot\)-stable basis,
  then the original topology on \(X\) and the apartness topology on \(X\)
  coincide, i.e.\ a subset of \(X\) is open (in the original topology) if and
  only if it is nearly open.
\end{lemma}

\begin{theorem}\label{apartness-Scott-topology}
  Let \(D\) be a continuous dcpo with a basis \(B\). Each of the following
  conditions on the basis \(B\) implies that \(\set{\upupset b \mid b \in B}\) is a
  \(\lnot\lnot\)-stable basis for the Scott topology on \(D\):
  \begin{enumerate}
  \item For every \(a,b\in B\), if \(\lnot\lnot(a \ll b)\), then \(\pa*{a \ll b}\).
  \item For every \(a,b\in B\), if \(\lnot\lnot(a \below b)\), then \(\pa*{a \below b}\).
  \item \(B\) satisfies \eqref{way-below-decidable}
  \item \(B\) satisfies \eqref{below-decidable}
  \end{enumerate}
  Hence, if one of these conditions holds, then the Scott topology on \(D\)
  coincides with the apartness topology on~\(D\) with respect to the intrinsic
  apartness, i.e.\ a subset of \(D\) is Scott open if and only if it is nearly
  open.
\end{theorem}

\begin{proof}
  The final claim follows from Lemma~\ref{matching-topologies-criterion} and the
  fact that the intrinsic apartness satisfies the reverse Kolmogorov property
  (by definition). Moreover, (iii) implies (i) and (iv) implies (ii). So it
  suffices to show that if (i) or (ii) holds, then
  \(\set{\upupset a \mid a \in B}\) is a \(\lnot\lnot\)-stable basis for the
  Scott topology on \(D\), viz.\ that
  \(\interior{\lnot\lnot\upupset a} \subseteq \upupset a\) for every
  \(a \in B\).  Let \(a \in B\) be arbitrary and suppose that
  \(x \in \interior{\lnot\lnot\upupset a}\).  Using Scott openness and
  continuity of \(D\), there exists \(b \in B\) such that \(b \ll x\) and
  \(b \in \lnot\lnot\upupset a\). The latter just says that
  \(\lnot\lnot (a \ll b)\).
  So if condition~(i) holds, then we get \(a \ll b\), so \(a \ll x\) and
  \(x \in \upupset a\), as desired.
  Now suppose that condition~(ii) holds. From \(\lnot\lnot(a \ll b)\), we get
  \(\lnot\lnot(a \below b)\) and hence, \(a \below b\) by condition~(ii). So
  \(a \below b \ll x\) and \(x \in \upupset a\), as wished.
\end{proof}

\section{Tightness, cotransitivity and sharpness}\label{sec:tightness-cotransitivity-sharpness}
\begin{definition}
  An apartness relation \(\apartvar\) on a set \(X\) is \emph{tight} if
  \(\lnot(x \apartvar y)\) implies \(x = y\) for every \(x,y \in
  X\). The~apartness is \emph{cotransitive} if \(x\apartvar y\) implies the
  disjunction of \(x \apartvar z\) and \(x \apartvar y\) for every \(x,y,z \in X\).
\end{definition}

\begin{lemma}\label{tight-implies-notnot-sep}
  If \(X\) is a set with a tight apartness, then \(X\) is
  \(\lnot\lnot\)-separated, viz.\ \(\lnot\lnot(x=y)\) implies \(x = y\) for
  every \(x,y \in X\).
\end{lemma}

It will be helpful to employ the following positive (but classically equivalent)
formulation of \(\pa*{x \below y} \land \pa*{x \neq y}\) from
\cite[Definition~20]{deJongEscardo2021b}.

\begin{definition}
  An element \(x\) of a dcpo \(D\) is \emph{strictly below} an element \(y\),
  written \(x \sbelow y\), if \(x \below y\) and for every \(z \aboveorder y\)
  and proposition \(P\), the equality
  \(z = \dirsup\pa*{\set{x} \cup \set{z \mid P}}\) implies \(P\).
\end{definition}

We can relate the above notion to the intrinsic apartness, but, although we do
not have a counterexample, we believe \(x \sbelow y\) to be weaker than
\(x \apart y\) in general.
\begin{proposition}\label{apart-implies-sbelow}
  If \(x \below y\) are elements of a dcpo \(D\) with \(x \apart y\), then
  \(x \sbelow y\).
\end{proposition}

\begin{example}
  In the Sierpi\'nski domain \(\Sierp\) we have \(\bot \sbelow \top\). In the
  powerset \(\powerset{X}\) of some set \(X\), the empty set is strictly below a
  subset \(A\) of \(X\) if and only if \(A\) is inhabited. More generally, if
  \(A \subseteq B\) are subsets of some set, then \(A \sbelow B\) holds if
  \(B \setminus A\) is inhabited, and if \(A\) is a decidable subset and
  \(A \sbelow B\), then \(B \setminus A\) is inhabited.
\end{example}

The following shows that we cannot expect tight or cotransitive apartness
relations on nontrivial dcpos.

\begin{theorem}\label{tight-cotransitive-wem}
  Let \(D\) be a dcpo with an apartness relation \({\apartvar}\).
  \begin{enumerate}
  \item If \(D\) has elements \(x \sbelow y\), then tightness of \({\apartvar}\)
    implies excluded middle.
  \item If \(D\) has elements \(x \below y\) with \(x \apartvar y\), then
    cotransitivity of \({\apartvar}\) implies weak excluded middle.
  \end{enumerate}
\end{theorem}

\begin{proof}
  (i): If \({\apartvar}\) is tight, then \(D\) is \(\lnot\lnot\)-separated by
  Lemma~\ref{tight-implies-notnot-sep}, which, since \(D\) has elements
  \(x \sbelow y\), implies excluded middle by
  \cite[Theorem~38]{deJongEscardo2021b}.

  (ii): For any proposition \(P\), consider the supremum \(s_P\) of the
  directed subset \(\set{x} \cup \set{y \mid P}\). If \({\apartvar}\) is
  cotransitive, then either \(x \apartvar s_P\) or \(y \apartvar s_P\). In the
  first case, \(x \neq s_P\), so that \(\lnot\lnot P\) must be the case. In the
  second case, \(y \neq s_P\), so that \(\lnot P\) must be true. Hence,
  \(\lnot P\) is decidable and weak excluded middle follows.
\end{proof}

\begin{theorem}\label{tight-cotransitive-em}\hfill
  \begin{enumerate}
  \item If excluded middle holds, then the intrinsic apartness on any dcpo is
    tight.
    In the other direction, if \(D\) is a continuous dcpo with elements
    \(x \below y\) that are intrinsically apart, then tightness of the intrinsic
    apartness on \(D\) implies excluded middle.
  \item If excluded middle holds, then the intrinsic apartness on any dcpo is
    cotransitive.
    In the other direction, if \(D\) is a continuous dcpo and has elements
    \(x \below y\) that are intrinsically apart, then cotransitivity of the
    intrinsic apartness on \(D\) implies excluded middle.
  \item In particular, if the intrinsic apartness on the Sierpi\'nski domain
    \(\Sierp\) is tight or cotransitive, then excluded middle follows.
  \end{enumerate}
\end{theorem}

We now isolate a collection of elements, which we call sharp elements, for which
the intrinsic apartness \emph{is} tight and cotransitive. The definition of a
sharp element of a continuous dcpo may be somewhat opaque, but the algebraic
case (Proposition~\ref{sharp-algebraic}) is easier to understand: an element
\(x\) is sharp if and only if \(c \below x\) is decidable for every compact
element \(c\). Sharpness also occurs naturally in our examples in
Section~\ref{sec:examples}, e.g.\ a lower real is sharp if and only if it is
located.

\begin{definition} An element \(x\) of a dcpo \(D\) is \emph{sharp} if for every
  \(y,z\in D\) with \(y \ll z\) we have \(y \ll x\) or \(z \not\below x\).
\end{definition}
Theorem~\ref{basis-is-sharp} below provides many examples of sharp elements.
Our first result is that sharpness is equivalent in continuous dcpos to a
seemingly stronger condition.

\begin{proposition}
  An element \(x\) of a continuous dcpo \(D\) is sharp if and only if for every
  \(y,z \in D\) with \(y \ll z\) we have \(y \ll x\) or \(z \posnotbelow x\).
\end{proposition}

\begin{lemma}\label{sharp-basis}
  An element \(x\) of a continuous dcpo \(D\) with a basis \(B\) is sharp if and
  only if for every \(a,b \in B\) with \(a \ll b\) we have \(a \ll x\) or
  \(b \not\below x\).
\end{lemma}

\begin{proposition}\label{sharp-algebraic}
  An element \(x\) of an algebraic dcpo \(D\) is sharp if and only if for every
  compact \(c \in D\) it is decidable whether \(c \below x\) holds.
\end{proposition}

\begin{proposition}\label{em-sharp}
  Assuming excluded middle, every element of any dcpo is sharp.
  The sharp elements of the Sierpi\'nski domain \(\Sierp\) are exactly \(\bot\)
  and \(\top\).
  Hence, if every element of \(\Sierp\) is sharp, then excluded middle
  follows.
\end{proposition}

We can relate the notion of sharpness to Spitters'~\cite{Spitters2010} and
Kawai's~\cite[Definition~3.5]{Kawai2017} notion of a located subset:
A subset \(V\) of a poset \(S\) is \emph{located} if for every \(s,t \in S\)
with \(s \ll t\), we have \(t \in V\) or \(s \not\in V\).

\begin{proposition}\label{sharp-iff-located-nbhds}
  An element \(x\) of a continuous dcpo is sharp if and only if the filter of
  Scott open neighbourhoods of \(x\) is located in the poset of Scott opens of
  \(D\).
\end{proposition}

The following gives examples of sharp elements.
\begin{theorem}\label{basis-is-sharp}
  Let \(D\) be a continuous dcpo with a basis \(B\).
  \begin{enumerate}
  \item Assuming that \(D\) is pointed, the least element of \(D\) is sharp if
    \(B\) satisfies \eqref{equals-bot-decidable}.
  \item Every element of \(B\) is sharp if \(B\) satisfies
    \eqref{way-below-decidable} or \eqref{below-decidable}.
    In particular, in these cases, the sharp elements are a Scott dense subset
    of \(D\) in the sense of Proposition~\ref{basis-is-dense}.
  \end{enumerate}

  If \(D\) is algebraic, then we can reverse the implications above for the
  basis of compact elements of \(D\).
  \begin{enumerate}[resume]
  \item Assuming that \(D\) is pointed, the least element of \(D\) is sharp if
    and only if the set of compact elements of \(D\) satisfies
    \eqref{equals-bot-decidable}.
  \item The compact elements of \(D\) are sharp if and only if the set of
    compact elements of \(D\) satisfies \eqref{way-below-decidable}
    or~\eqref{below-decidable}.
  \end{enumerate}
\end{theorem}

\begin{theorem}\label{sharp-tight-cotransitive}\hfill
  \begin{enumerate}
  \item If \(y\) is a sharp element of a continuous \(D\), then
    \(\lnot (x \posnotbelow y)\) implies \(x \below y\) for every \(x \in D\).
    In particular, the intrinsic apartness on a continuous dcpo \(D\) is tight
    on sharp elements.
  \item The intrinsic apartness on a continuous dcpo \(D\) is cotransitive with
    respect to sharp elements in the following sense: for every \(x,y \in D\)
    and sharp element \(z \in D\), we have
    \(x \apart y \to \pa*{x \apart z \vee y \apart z}\).
  \end{enumerate}
\end{theorem}

\section{Strongly maximal elements}\label{sec:strongly-maximal-elements}
Smyth~\cite{Smyth2006} explored the notion of a \emph{constructively maximal}
element, adapted from Martin-L\"of's \cite{MartinLof1970}. In an unpublished
manuscript~\cite{Heckmann1998}, Heckmann arrived at an equivalent notion,
assuming excluded middle, as noted in \cite[Section~8]{Smyth2006}, and called it
\emph{strong maximality}. Whereas Smyth works directly with abstract bases and
rounded ideal completions, we instead work with continuous dcpos. We use a
simplification of Smyth's definition, but follow Heckmann's terminology. We show
that every strongly maximal element is sharp, compare strong maximality to
maximality highlighting connections to sharpness, and study the subspace of
strongly maximal elements.

\begin{definition}
  Two points \(x\) and \(y\) of a dcpo \(D\) are \emph{Hausdorff separated} if
  we have disjoint Scott open neighbourhoods of \(x\) and \(y\)
  respectively.
\end{definition}

\begin{definition}\label{def:strongly-maximal}
  An element \(x\) of a continuous dcpo \(D\) is \emph{strongly maximal} if for
  every \(u,v \in D\) with \(u \ll v\), we have \(u \ll x\) or \(v\) and \(x\)
  are Hausdorff separated.
\end{definition}

The following gives another source of examples of sharp elements.
\begin{proposition}\label{sharp-if-strongly-maximal}
  Every strongly maximal element of a continuous dcpo is sharp.
\end{proposition}

\begin{corollary}\hfill
  \begin{enumerate}
  \item The intrinsic apartness on a continuous dcpo \(D\) is tight on strongly
    maximal elements.
  \item The intrinsic apartness on a continuous dcpo \(D\) is cotransitive with
    respect to strongly maximal elements in the following sense: for every
    \(x,y \in D\) and strongly maximal element \(z \in D\), we have
    \({x \apart y \to \pa*{x \apart z \vee y \apart z}}\).
  \end{enumerate}
\end{corollary}

\begin{lemma}\label{strongly-maximal-basis}
  If a continuous dcpo \(D\) has a basis \(B\), then an element \(x \in D\) is
  strongly maximal if~and~only~if for every \(a,b \in B\) with \(a \ll b\),
  we have \(a \ll x\) or \(b\) and \(x\) are Hausdorff separated.
\end{lemma}

\begin{lemma}\label{strongly-maximal-algebraic}
  An element \(x\) of an algebraic dcpo \(D\) is strongly maximal if and only if
  for every compact element \(c \in D\) either \(c \below x\) or \(c\) and \(x\) are
  Hausdorff separated.
\end{lemma}

Smyth's formulation of strong maximality~\cite[Definition~4.1]{Smyth2006}
(called \emph{constructive maximality} there) is somewhat more involved than
ours, but it is equivalent, as Proposition~\ref{Smyth-strong-equiv} shows.

\begin{definition}\label{def:refine}
  We say that two elements \(x\) and \(y\) of a dcpo \(D\) can be
  \emph{refined}, written \(x \refined y\), if there exists \(z \in D\) with
  \(x \ll z\) and \(y \ll z\).
\end{definition}
On~\cite[p.~362]{Smyth2006}, refinement is denoted by
\(x \mathrel{\uparrow} y\), but we prefer \(x \refined y\), because one might
want to reserve \(x \mathrel{\uparrow} y\) for the weaker
\(\exists_{z \in D}\,\pa*{\pa*{x \below z} \land \pa*{y \below z}}\).

\begin{lemma}\label{Hausdorff-separated-criterion}
  Two elements \(x\) and \(y\) of a continuous dcpo \(D\) are Hausdorff
  separated if and only if there exist
  \(a,b \in D\) with \(a \ll x\) and \(b \ll y\) such that
  \(\lnot \pa*{a \refined b}\).
  Moreover, if \(D\) has a basis \(B\), then \(x\) and \(y\) are Hausdorff
  separated if and only if there exist such elements \(a\) and \(b\) in \(B\).
\end{lemma}

In \cite[p.~362]{Smyth2006}, the condition in
Lemma~\ref{Hausdorff-separated-criterion} is taken as a definition (for basis
elements) and such elements \(a\)~and~\(b\) are said to lie apart and this
notion is denoted by \(a \mathrel{\sharp} b\).
We can now translate Smyth's~\cite[Definition~4.1]{Smyth2006} from ideal
completions of abstract bases to continuous dcpos.
\begin{definition}
  An element \(x\) of a continuous dcpo \(D\) is \emph{Smyth maximal} if for
  every \(u,v \in D\) with \(u \ll v\), there exists \(d \ll x\) such that
  \(u \ll d\) or the condition in Lemma~\ref{Hausdorff-separated-criterion}
  holds for \(v\) and \(d\).
\end{definition}

\begin{proposition}\label{Smyth-strong-equiv}
  An element \(x\) of a continuous \(D\) is strongly maximal if and only if
  \(x\) is Smyth maximal.
\end{proposition}

\subsection{Maximality and strong maximality}\label{subsec:max-strong-max}
The name strongly maximal is justified by the following observation.
\begin{proposition}[cf.\ Proposition~4.2 in \cite{Smyth2006}]\label{maximal-if-strongly-maximal}
  Every strongly maximal element of a continuous dcpo is maximal.
\end{proposition}

In the presence of excluded middle, \cite[Corollary~4.4]{Smyth2006} tells us
that the converse of the above is true if and only if the
{Lawson~condition}~\cite{Smyth2006,Lawson1997} holds for the dcpo (the Scott and
Lawson topologies coincide on the subset of maximal elements.)
Without excluded middle, the situation is subtler and involves sharpness, as
we proceed to show.

\begin{lemma}\label{maximal-strongly-maximal-wem}
  Suppose that we have elements \(x\) and \(y\) of a continuous dcpo \(D\) such
  that
  \begin{enumerate}
  \item\label{both-strongly-max} \(x\) and \(y\) are both strongly maximal;
  \item\label{glb} \(x\) and \(y\) have a greatest lower bound \(x \glb y\) in \(D\);
  \item\label{apart} \(x\) and \(y\) are intrinsically apart.
  \end{enumerate}
  Then, for any proposition \(P\), the supremum \(\dirsup S\) of the directed
  subset
  \(S \coloneqq \set{x \glb y} \cup \set{x \mid \lnot P} \cup \set{y \mid
    \lnot\lnot P}\) is maximal, but \(\dirsup S\) is sharp if and only if
  \(\lnot P\) is decidable.  Hence, if \(\dirsup S\) is strongly maximal, then
  \(\lnot P\) is decidable.
\end{lemma}

\begin{proposition}\label{maximal-strongly-maximal-wem-alt}
  Let \(P\) be the poset with exactly three elements \(\bot \leq 0,1\) and \(0\)
  and \(1\) unrelated.  If every maximal element of the algebraic dcpo
  \(\Idl(P)\) is strongly maximal, then weak excluded middle follows.
  In the other direction, if excluded middle holds, then every maximal elements
  of \(\Idl(P)\) is strongly maximal.
\end{proposition}

Theorem~4.3 of \cite{Smyth2006} states that an element \(x\) of a continuous
dcpo is strongly maximal if and only if every Lawson neighbourhood of \(x\)
contains a Scott neighbourhood of \(x\).
This requirement on neighbourhoods is, assuming excluded middle, equivalent to
the Lawson condition (the Scott topology and the Lawson topology coincide on the
subset of maximal elements), as used in~\cite[Corollary~4.4]{Smyth2006} and
proved in \cite[Lemma~V\nobreakdash-6.5]{GierzEtAl2003}. Inspecting the proof
in~\cite{GierzEtAl2003}, we believe that excluded middle is essential. However,
we can still prove a constructive analogue of \cite[Theorem~4.3]{Smyth2006}, but
it requires a positive formulation of the subbasic opens of the Lawson topology,
using the relation \({\posnotbelow}\) rather than \({\not\below}\).

\begin{definition}\label{Lawson-subbasics}
  The subbasic \emph{Lawson closed} subsets of a dcpo \(D\) are the Scott closed
  subsets and the upper sets of the form \(\upset x\) for \(x \in D\).
  The subbasic \emph{Lawson opens} are the Scott opens and the sets of
  the form \(\set{y \in D \mid x \posnotbelow y}\) for \(x \in D\).
\end{definition}

In the presence of excluded middle, the subset
\(\set{y \in D \mid x \posnotbelow y}\) is equal to \(D \setminus \upset x\), so
with excluded middle the above definition is equivalent to the classical
definition of the subbasic Lawson opens as
in~{\cite[pp.~211--212]{GierzEtAl2003}}.

\begin{theorem}[cf.\ Theorem~4.3 of \cite{Smyth2006}]
  \label{Lawson-strongly-maximal}
  An element \(x\) of a continuous dcpo is strongly maximal if and only if \(x\)
  is sharp and every Lawson neighbourhood of \(x\) contains a Scott
  neighbourhood of \(x\).
\end{theorem}

By Proposition~\eqref{em-sharp}, every element is sharp if excluded middle is
assumed, so in that case, we can drop the requirement that \(x\) is
sharp. Hence, in the presence of excluded middle, we recover
Smyth's~\cite[Theorem~4.3]{Smyth2006} from
Theorem~\ref{Lawson-strongly-maximal}.

\subsection{The subspace of strongly maximal elements}
The classical interest in strong maximality comes from the fact that, while the
subspace of maximal elements may fail to be
Hausdorff~\cite[Example~4]{Heckmann1998}, the subspace of strongly maximal
elements with the relative Scott topology is both Hausdorff and
regular~\cite[Theorem~4.6]{Smyth2006}. We offer constructive proofs of these
claims, with the proviso that the Hausdorff condition is formulated with respect
to the intrinsic apartness.
\begin{proposition}
  The subspace of strongly maximal elements of a continuous dcpo \(D\) with the
  relative Scott topology is Hausdorff, i.e.\ if \(x\)~and~\(y\) are strongly
  maximal, then \(x \apart y\) if and only if there are disjoint Scott opens
  \(U\) and \(V\) such that \(x \in U\) and \(y \in V\).
\end{proposition}

\begin{proposition}
  The subspace of strongly maximal elements of a continuous dcpo \(D\) with the
  relative Scott topology is regular, i.e.\ every Scott neighbourhood of a point
  \(x \in D\) contains a Scott closed neighbourhood of \(x\).
\end{proposition}

\section{Examples}\label{sec:examples}
In this final section before the conclusion, we illustrate the foregoing
notions of intrinsic apartness, sharpness and strong maximality, by studying
several natural examples. The first three examples are generalized domain
environments in the sense of Heckmann~\cite{Heckmann1998}: we consider dcpos
\(D\) and topological spaces \(X\) such that \(X\) \emph{embeds} into the
subspace of maximal elements of \(D\). In fact, we will see that \(X\) is
homeomorphic to the subspace of \emph{strongly} maximal elements of \(D\).
Specifically, we will consider Cantor space, Baire space and the real line.
The final example shows that sharpness characterizes exactly those lower reals
that are located.

\subsection{The Cantor and Baire domains}\label{sec:Cantor-and-Baire-domains}
Fix an inhabited set \(A\) with decidable equality. Typically, we will be
interested in \(A = \Two \coloneqq \set{0,1}\) and \(A = \Nat\).

\begin{definition}\hfill
  \begin{enumerate}
  \item Write \(\pa*{\Finseq , \preceq}\) for the poset of finite sequences on
    \(A\) ordered by prefix. We write \(\Seqdom\) for the ideal completion of
    \(\pa*{\Finseq , \preceq}\), which is an algebraic dcpo by
    Lemma~\ref{reflexive-algebraic}.
  \item For an infinite sequence \(\alpha\), we write \(\initseg{\alpha}{n}\)
    for the first \(n\) elements of \(\alpha\).  Given a finite
    sequence~\(\sigma\), we write \(\sigma \prec \alpha\) if \(\sigma\) is an
    initial segment of \(\alpha\).
  \item We write \(\Seqspace\) for the space \(A^\Nat\) of infinite sequences on
    \(A\) with the product topology, taking the discrete topologies on \(\Nat\)
    and \(A\). A basis of opens is given by the sets
    \(\set{\alpha \in \Seqspace \mid \sigma \prec \alpha}\) for finite sequences
    \(\sigma\). The space \(\Seqspace\) has a natural notion of apartness:
    \(\alpha \apart_{\Seqspace} \beta \iff \exists_{n \in
      \Nat}\,\initseg{\alpha}{n} \neq \initseg{\beta}{n}\).
  \end{enumerate}
\end{definition}

\begin{definition}
  If we take \(A = \Two\), then \(\Seqspace\) is Cantor space and we call
  \(\Seqdom\) the \emph{Cantor domain}; and if we take \(A = \Nat\), then
  \(\Seqspace\) is Baire space and we call \(\Seqdom\) the \emph{Baire domain}.
\end{definition}

\begin{definition}
  We define an injection \(\iota \colon \Seqspace \hookrightarrow \Seqdom\) by
  \(\iota(\alpha) \coloneqq \bigcup_{\sigma\prec\alpha} \dset\sigma = \set{\tau
    \in \Finseq \mid \tau \prec \alpha}\).
\end{definition}

\begin{theorem}\label{seq-iota-strongly-maximal}
  The image of \(\iota\) is exactly the subset of strongly maximal elements of
  \(\Seqdom\).
\end{theorem}

\begin{proposition}\label{seq-strongly-max-wem}
  Suppose that \(A\) has at least two elements. If every
  maximal element of \(\Seqdom\) is strongly maximal, then weak excluded middle
  holds.
\end{proposition}

\begin{definition}
  \emph{Markov's Principle} is the assertion that for every infinite binary
  sequence \(\phi\) we have that
  \(\lnot \pa*{\forall_{n \in \Nat}\,\phi(n) = 0}\) implies
  \(\exists_{k \in \Nat}\,\phi(k) =1\).
\end{definition}
Markov's Principle~\cite{BridgesRichman1987} follows from excluded middle, but
is independent in constructive mathematics, i.e.\ Markov's Principle is not
provable and neither is its negation.

\begin{theorem}
  For the Cantor domain, the strongly maximal elements are exactly those
  elements that are both sharp and maximal.
  For the Baire domain, the strongly maximal elements are exactly the elements
  that are both sharp and maximal if and only if Markov's Principle holds.
\end{theorem}

\begin{lemma}\label{seq-T0}
  The space \(\Seqspace\) is \(T_0\)-separated with respect to
  \(\apart_{\Seqspace}\), i.e.\ for \(\alpha,\beta \in \Seqspace\) we have
  \(\alpha \apart_{\Seqspace} \beta\) if and only if there exists an open
  containing \(x\) but not \(y\) or vice versa.
\end{lemma}

\begin{theorem}\label{seq-homeo-apart}
  The map \(\iota\) is a homeomorphism from the space \(\Seqspace\) of infinite
  sequences to the space of strongly maximal elements of the algebraic dcpo
  \(\Seqdom\) with the relative Scott topology.
  Moreover, \(\iota\) preserves and reflects apartness:
  \(\alpha \apart_{\Seqspace} \beta\) if and only if
  \(\iota(\alpha) \apart_{\Seqdom} \iota(\beta)\) for every two infinite
  sequences \(\alpha,\beta \in \Seqspace\).
\end{theorem}

\subsection{Partial Dedekind reals}
Recall the definition of a (two-sided) Dedekind real number.
\begin{definition}\label{def:Dedekind-real}
  Given a pair \(x = (L_x,U_x)\) of subsets of \(\Rat\), we write \(p < x\) for
  \(p \in L_x\) and \(x < q\) for \(q \in U_x\). A~\emph{Dedekind real} \(x\) is
  a pair \((L_x,U_x)\) of subsets of \(\Rat\) satisfying the following properties:
  \begin{enumerate}
  \item \emph{boundedness}: there exist \(p \in \Rat\) and \(q \in \Rat\) such
    that \(p < x\) and \(x < q\).
  \item \emph{roundedness}: for every \(p \in \Rat\), we have
    \(p < x \iff \exists_{r \in \Rat} \pa*{p < r} \land \pa*{r < x}\) and
    similarly, for every \(q \in \Rat\), we have
    \(x < q \iff \exists_{s \in \Rat}\pa{s < q} \land \pa{x < s}\).
  \item \emph{transitivity}: for every \(p,q \in \Rat\), if \(p < x\) and
    \(x < q\), then \(p < q\).
  \item \emph{locatedness}: for every \(p,q\in\Rat\) with \(p < q\) we have
    \(p < x\) or \(x < q\).
  \end{enumerate}
\end{definition}

\begin{definition}
  The topological space \(\Rea\) is the set of all Dedekind real numbers whose
  basic opens are given by \(\set{x \in \Rea \mid p < x \text{ and } x < q}\)
  for \(p,q \in \Rat\). The space \(\Rea\) has a natural notion of apartness, namely:
  \(x \apart_{\Rea} y \iff {\exists_{p \in \Rat}\,\pa*{x < p < y} \vee \pa*{y
      < p < x}}\).
\end{definition}

\begin{definition}
  Consider the set
  \(\Rat \times_{<} \Rat \coloneqq \set{(p,q) \in \Rat\times\Rat \mid p < q}\)
  ordered by defining the strict order \((p,q) \prec (r,s) \iff p < r < s <
  q\). The pair \(\pa*{\Rat \times_{<} \Rat , \prec}\) is an abstract basis, so
  \(\Realdom \coloneqq \Idl(\Rat\times_{<}\Rat,\prec)\) is a continuous dcpo and
  we refer to its elements as \emph{partial Dedekind reals}.
\end{definition}

\begin{lemma}\label{Dedekind-way-below-criterion}
  For every two rationals \(p < q\) and \(I \in \Realdom\), we have
  \(\dset(p,q) \ll I\) if and only if \((p,q) \in I\).
\end{lemma}

\begin{definition}
  We define an injection \(\iota \colon \Rea \hookrightarrow \Realdom\) by
  \(\iota\pa*{L_x,U_x} \coloneqq \set{(p,q) \mid p \in L_x, q \in U_x}\).
  The map \(\iota\) is well-defined precisely because a Dedekind real is
  required to be bounded, rounded and transitive.
\end{definition}

\begin{theorem}\label{Dedekind-iota-strongly-maximal}
  The image of \(\iota\) is exactly the subset of strongly maximal elements of
  \(\Realdom\).
\end{theorem}

With excluded middle, the image of \(\iota\) is just the set of maximal elements
of \(\Realdom\). The following result highlights the constructive strength of
locatedness of Dedekind reals.

\begin{proposition}\label{Dedekind-max-strong-max-wem}
  If every maximal element of \(\Realdom\) is strongly maximal, then weak
  excluded middle holds.
\end{proposition}

We conjecture that \(\Realdom\) is similar to the Baire domain in that the
strongly maximal elements of \(\Realdom\) only coincide with the elements that
are both sharp and maximal if a constructive taboo holds.

\begin{lemma}\label{Dedekind-T0}
  The Dedekind real numbers are \(T_0\)-separated with respect to \(\apart_{\Rea}\),
  i.e.\ for \(x,y \in \Rea\) we have
  \(x \apart_{\Rea} y\) if and only if there exists an open \(U\)
  containing \(x\) but not \(y\) or vice versa.
\end{lemma}

\begin{theorem}
  The map \(\iota\) is a homeomorphism from \(\Rea\) to the space of strongly
  maximal elements of the continuous dcpo \(\Realdom\) with the relative Scott
  topology.
  Moreover, \(\iota\) preserves and reflects apartness.
\end{theorem}

\subsection{Lower reals}
We now consider lower reals, which feature a nice illustration of sharpness.
\begin{definition}
  The pair \(\pa*{\Rat , <}\) is an abstract basis, so
  \(\LowerReal \coloneqq \Idl(\Rat,<)\) is a continuous dcpo and we refer to its
  elements as \emph{lower reals}.
\end{definition}

\begin{lemma}\label{lower-reals-way-below-criterion}
  For every \(p \in \Rat\) and \(L \in \LowerReal\), we have
  \(\dset p \ll L\) if and only if \(p \in L\).
\end{lemma}

\begin{lemma}\label{U-from-L}
  If \(L \in \LowerReal\) is a lower real, then the pair \((L,U)\) with
  \(U \coloneqq \set*{q \in \Rat \mid \exists_{s \in \Rat \setminus L}\,s < q}\)
  is rounded and transitive in the sense
  of~Definition~\ref{def:Dedekind-real}.
  Moreover, if \(\Rat \setminus L\) is inhabited, then \((L,U)\) is bounded too.
\end{lemma}

Classically, every lower real whose complement is inhabited determines a
Dedekind real by the construction above. It is well-known that constructively a
lower real may fail to be located. The following result offers a
domain-theoretic explanation of that phenomenon.

\begin{theorem}\label{sharp-iff-located}
  A lower real \(L \in \LowerReal\) is sharp if and only if the pair \((L,U)\)
  with \(U\)~as in Lemma~\ref{U-from-L} is located.
\end{theorem}

\nocite{Escardo2008}

\section{Conclusion}\label{sec:conclusion}
Working constructively, we studied continuous dcpos and the Scott topology and
introduced notions of intrinsic apartness and sharp elements.
We showed that our apartness relation is particularly well-suited for continuous
dcpos that have a basis satisfying certain decidability conditions, which hold in
examples of interest. For instance, for such continuous dcpos, the
Bridges--{V\^i\c{t}\v{a}} apartness topology and the Scott topology coincide.
We proved that no apartness on a nontrivial dcpo can be cotransitive
or tight unless (weak) excluded middle holds. But the intrinsic apartness is
tight and cotransitive when restricted to sharp elements.
If a continuous dcpos has a basis satisfying the previously mentioned
decidability conditions, then every basis element is sharp. Another class of
examples of sharp elements is given by the strongly maximal elements.
In fact, strong maximality is closely connected to sharpness and the Lawson
topology. For example, an element \(x\) is strongly maximal if and only if \(x\)
is sharp and every Lawson neighbourhood of~\(x\) contains a Scott neighbourhood
of~\(x\).
Finally, we presented several natural examples of continuous dcpos that
illustrated the intrinsic apartness, strong maximality and sharpness.

In future work, we would like to explore whether a constructive and predicative
treatment is possible, in particular, in univalent foundations without
Voevodsky's resizing axioms as in~\cite{deJongEscardo2021a}. Steve Vickers also
pointed out two directions for future research. The first is to consider formal
ball domains~\cite[Example~V-6.8]{GierzEtAl2003}, which may subsume the partial
Dedekind reals example. The second is to explore the ramifications of Vickers'
observation that refinability (Definition~\ref{def:refine}) is decidable, even
when the order is not, if the dcpo is algebraic and
2/3~SFP~\cite[p.~157]{Vickers1989}.
Related to the examples, there is still the question of whether we can derive a
constructive taboo from the assumption that strong maximality of a partial
Dedekind real follows from having both sharpness and maximality, as discussed
right after Proposition~\ref{Dedekind-max-strong-max-wem}.
Finally, the Lawson topology deserves further investigation within a
constructive framework.

\section*{Acknowledgements}
I am very grateful to Mart\'in Escard\'o for many discussions (including one
that sparked this paper) and valuable suggestions for improving the
exposition. In particular, the terminology ``sharp'' is due to Mart\'in and
Theorem~\ref{sharp-iff-located} was conjectured by him.
I~should also like to thank Steve Vickers for his interest and remarks.
Finally, I thank the anonymous referees for their helpful comments and
questions.

\bibliographystyle{eptcs}
\bibliography{references}

\appendix
\section*{Appendix}
For lack of space, we have omitted a number of proofs. A version of this paper
containing full proofs can be found here:
\href{https://arxiv.org/abs/2106.05064}{\texttt{arXiv:2106.05064}}.
It also contains a section on an additional example, namely an alternative
domain for Cantor space using exponentials and the lifting monad.

\end{document}